\documentclass{techrep}

\usepackage{amsfonts}
\usepackage{amsthm}

\newcommand{\comment}[1]{}

\newtheorem{theorem}{Theorem}{\bfseries}{\itshape}

{\bfseries}{\itshape}

\newtheorem{prop}{Proposition}{\bfseries}{\itshape}

{\bfseries}{\itshape}

{\bfseries}{\itshape}

\newtheorem{obs}{Observation}{\bfseries}{\itshape}

{\bfseries}{\itshape}

\newtheorem{assumption}{Assumption}{\bfseries}{\rm}

\theoremstyle{definition}

{\bfseries}{\rm}

\theoremstyle{plain}

\usepackage{verbatim}
\usepackage{algorithm2e}
\usepackage[utf8]{inputenc}
\usepackage{microtype}
\usepackage{amsthm,amsmath,graphicx}
\usepackage{caption}
\usepackage{subcaption}
\usepackage[letterpaper]{hyperref}
\usepackage[table,dvipsnames]{xcolor}
\definecolor{linkblue}{named}{Blue}
\hypersetup{colorlinks=true, linkcolor=linkblue,  anchorcolor=linkblue,
	citecolor=linkblue, filecolor=linkblue, menucolor=linkblue,
	urlcolor=linkblue} 
\usepackage{wasysym}
\usepackage{graphicx}
\usepackage{caption}
\usepackage{enumitem}
\usepackage{thmtools, thm-restate}
\usepackage{wrapfig}
\usepackage{float}
\usepackage{stackrel}
\usepackage{fixfoot}
\usepackage[utf8]{inputenc}
\usepackage[english]{babel}
\usepackage{mathrsfs}  
\usepackage{venturis}
\usepackage{placeins}

\newlength\problemsep
\title{Tight Routing and Spanning Ratios of Arbitrary Triangle Delaunay Graphs}

\author{Prosenjit Bose, Jean-Lou De Carufel, John Stuart}

\begin{document}

\date{\today}

\maketitle
\begin{abstract}
 A Delaunay graph built on a planar point set has an edge between two vertices when there exists a disk with the two vertices on its boundary and no vertices in its interior. When the disk is replaced with an equilateral triangle, the resulting graph is known as a \emph{Triangle-Distance Delaunay Graph} or TD-Delaunay for short. A generalized $\text{TD}_{\theta_1,\theta_2}$-Delaunay graph is a TD-Delaunay graph whose empty region is a scaled translate of a triangle with angles of $\theta_1,\theta_2,\theta_3:=\pi-\theta_1-\theta_2$  with $\theta_1\leq\theta_2\leq\theta_3$. We prove that $\frac{1}{\sin(\theta_1/2)}$ is a lower bound on the spanning ratio of these graphs which matches the best known upper bound (Lubiw \& Mondal \ \emph{J. Graph Algorithms Appl.}, \textbf{23}(2):345--369). Then we provide an online local routing algorithm for $\text{TD}_{\theta_1,\theta_2}$-Delaunay graphs with a routing ratio that is optimal in the worst case. When $\theta_1=\theta_2=\frac{\pi}{3}$, our expressions for the spanning ratio and routing ratio evaluate to $2$ and $\frac{\sqrt{5}}{3}$, matching the known tight bounds for TD-Delaunay graphs.
\end{abstract}

\section{Introduction} 
Geometric graphs are graphs whose vertex sets are points in the plane and whose edge weights are the corresponding Euclidean distances. A common theme in Computational Geometry is the study of shortest paths. In geometric graphs, one measure of how well a graph preserves distances is its spanning ratio. The spanning ratio of a geometric graph is the smallest upper bound on the ratio of distance in the graph to distance in the plane for all pairs of points \cite{narasimhan_smid_2007}. One particular geometric graph of interest is the Delaunay triangulation, which has an edge between two vertices exactly when they lie on the boundary of a disk which contains no other vertex in its interior \cite{DBLP:books/lib/BergCKO08}. 

A long-standing open problem is to determine the worst-case spanning ratio of the Delaunay triangulation, which is known to be between $1.5932$ \cite{DBLP:conf/cccg/XiaZ11} and $1.998$ \cite{DBLP:journals/siamcomp/Xia13}. In other words, there exists a point set where the spanning ratio is at least $1.5932$, and for any point set, the spanning ratio is at most $1.998$. While the exact spanning ratio of the standard Delaunay triangulation remains unknown, several variants do have known tight spanning ratios in the worst case. For example, when the empty disk is replaced with a square we obtain the $L_\infty$ or $L_1$-Delaunay graph, which is known to have a spanning ratio of exactly $\sqrt{4+2\sqrt{2}}\approx 2.61$ \cite{DBLP:journals/comgeo/BonichonGHP15}. Similar proof techniques have been generalized to Delaunay graphs based on rectangles and parallelograms \cite{DBLP:conf/esa/RenssenSSW23,mythesis}. In general, one can define a Delaunay graph from any convex distance function, and such a graph is known to have a constant spanning ratio where the spanning ratio depends on the ratio of the perimeter to the width of the convex shape \cite{DBLP:journals/jocg/BoseCCS10}. When the unit circle in this distance is a regular hexagon, then the exact worst-case spanning ratio is $2$ \cite{DBLP:journals/jocg/Perkovic0T21}. When the unit circle is an equilateral triangle, then exact worst-case spanning ratio is also $2$ \cite{DBLP:journals/jcss/Chew89}. A generalized $\text{TD}_{\theta_1,\theta_2}$-Delaunay graph is a TD-Delaunay graph whose empty region is a scaled translate of a triangle with angles of $\theta_1,\theta_2,\theta_3:=\pi-\theta_1-\theta_2$ with $\theta_1\leq\theta_2\leq\theta_3$. In this paper, we provide a lower bound of $\frac{1}{\sin(\theta_1/2)}$ that matches the best known upper bound for $\text{TD}_{\theta_1,\theta_2}$-Delaunay graphs \cite{DBLP:journals/jgaa/LubiwM19}.

The routing ratio of a geometric graph essentially captures how feasible it is to find short paths in a graph when making local decisions based only on the neighbourhood of the current vertex. The routing ratio is the smallest upper bound on the ratio of the length of the path returned by the routing algorithm and the Euclidean distance between all pairs of vertices. Routing in Delaunay trianglulations is notoriously difficult, with the routing ratio of the standard Delaunay triangulation known to be between $1.70$ \cite{DBLP:journals/dcg/BonichonBCPR17} and $3.56$ \cite{DBLP:journals/dcg/BonichonBCPR17}. Variations such as the $L_1$-Delaunay triangulation are known to have a routing ratio between $2.7$ \cite{DBLP:journals/dcg/BonichonBCPR17} and $3.16$ \cite{DBLP:conf/compgeom/Chew86}. 

For TD-Delaunay graphs, there is a gap between the spanning ratio of $2$ and the routing ratio which was shown to be exactly $\frac{5}{\sqrt{3}}$ in the worst-case \cite{DBLP:journals/siamcomp/BoseFRV15}. We show that this gap is preserved for $\text{TD}_{\theta_1,\theta_2}$-Delaunay graphs by extending techniques from \cite{DBLP:journals/siamcomp/BoseFRV15} to obtain a tight routing ratio of 
\begin{align*}
    C(\theta_1,\theta_2):=\max_{\substack{j\in\{1,2,3\} \\0\leq \alpha\leq\theta_j}} \Bigg(\frac{\sin(\theta_{j}-\alpha)}{\sin(\theta_{j+1})}+\frac{\sin(\alpha)}{\sin(\theta_{j-1})}+\min\Big(&
    \frac{\sin(\alpha)}{\sin(\theta_{j-1})}+
    \frac{\sin(\alpha+\theta_{j-1})}{\sin(\theta_{j+1})},\\
    &\frac{\sin(\theta_{j}-\alpha)}{\sin(\theta_{j+1})}+\frac{\sin(\alpha+\theta_{j-1})}{\sin(\theta_{j-1})}\Big)\Bigg).
\end{align*}

\section{Preliminaries}

We will denote the line segment from point $u$ to point $v$ as $uv$, and the length of $uv$ is denoted $|uv|$. For two vertices $u,v$ in a geometric graph $G$, the length of the shortest path from $u$ to $v$ in $G$ is denoted $d_G(u,v)$. Then for a constant $c\geq 1$, $G$ is said to be a $c$-spanner if for all vertices $u,v$ in $G$, we have $d_G(u,v)\leq c|uv|$. The spanning ratio of $G$ is the least $c$ for which $G$ is a $c$-spanner. The spanning ratio of a class of graphs $\mathcal{G}$ is the least $c$ for which all graphs in $\mathcal{G}$ are $c$-spanners. A {\em constant} spanner is a $c$-spanner where $c$ is a constant.

In a geometric graph, each vertex is identified with its coordinates. Here, one unit of memory is either a point in $\mathbb{R}^2$, or $\log_2(n)$ bits. The $k$-neighbourhood of a vertex $u$ in a graph is defined to be all the vertices $v$ such that there is a path from $u$ to $v$ consisting of $k$ or fewer edges. Formally, a $k$-local, $m$-memory routing algorithm is a function that takes as input $(s,N_k(s),t,M)$, and outputs a vertex $p$ where $s$ is the current vertex, $N_k(s)$ is the $k$-neighbourhood of $s$, $t$ is the destination, $M$ is an $m$-unit memory register, and $p\in N_1(s)$. An algorithm is said to be $c$-competitive for a family of geometric graphs $\mathcal{G}$ if the path output by the algorithm for any pair of vertices $s,t\in V(G)$ for $G\in\mathcal{G}$ has length at most $c|st|$. The routing ratio of an algorithm is the least $c$ for which the algorithm is $c$-competitive for $\mathcal{G}$.

Throughout this paper, we fix a triangle $\triangle$ in the plane with angles $\theta_1\leq \theta_2\leq \theta_3$. We assume that the corresponding corners of $\triangle$ are labelled $\tau_1, \tau_2, \tau_3$. In order to keep notation cleaner, we use arithmetic modulo $3$ for operations on index $i$ when referring to corners of triangles. For example, $\tau_{4}=\tau_1$, and $\tau_0=\tau_3$. By convention, the expression $\angle abc$ will refer to the smaller angle among the clockwise and counterclockwise angles between $ab$ and $bc$ for three non-collinear points $a,b,c\in\mathbb{R}^2$.

For any two points $u,v$ in the plane, define the triangle $T^{u,v}$ to be the smallest scaled translate of $\triangle$ with $u$ and $v$ on its boundary. Note that \textit{smallest} implies that at least one of $u,v$ is on a corner of $T^{u,v}$. We use $\tau^{u,v}_i$ to refer to the corner of triangle $T^{u,v}$ corresponding to $\tau_i$. Now we define the cones, depicted in Figure \ref{fig:cones}. In particular, for a point $p$ and index $i\in\{1,2,3\}$, let $C_{p,i}:=\{v\in\mathbb{R}^2 | p=\tau^{p,v}_i\}$ be the positive cone centred at point $p$ corresponding to $\tau_i$. On the other hand, define the negative cone $\overline{C_{p,i}}:=\{v\in\mathbb{R}^2 | v=\tau^{p,v}_i\}$. Note that $\overline{C_{p,i}}$ is $C_{p,i}$ rotated by $\pi$ radians about $p$.

\begin{figure}[ht!]
    \centering
    \includegraphics[page=2, scale=0.7]{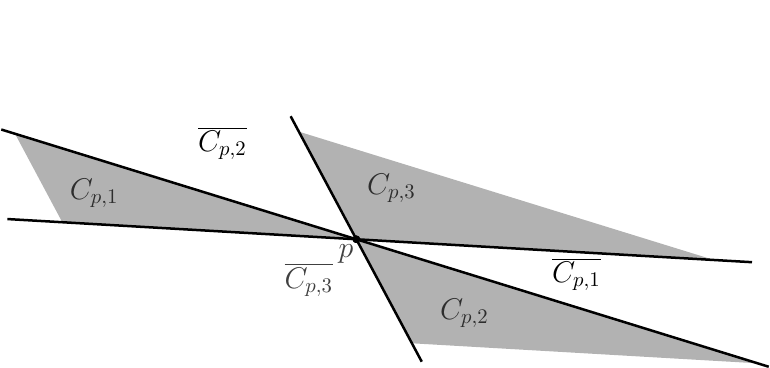}
    \caption{$C_{p,1},C_{p,2},C_{p,3}$ are the positive cones of $p$ and $\overline{C_{p,1}},\overline{C_{p,2}},\overline{C_{p,3}}$ are the negative cones of $p$.\label{fig:cones}}
\end{figure}

The TD-Delaunay graph of a vertex set $S\subseteq \mathbb{R}^2$ has an edge between vertices $u$ and $v$ when there exists an equilateral triangle with $u,v$ on its boundary and no other points of $S$ in its interior. Note that the equilateral triangle is a scaled translate of the TD unit circle. As with any Delaunay graph based on a convex distance function, every bounded face is a triangle \cite{DBLP:journals/jocg/BoseCCS10}. 
To define the $\text{TD}_{\theta_1,\theta_2}$-Delaunay graph, we replace the equilateral triangle with $\triangle$ containing angles of $\theta_1,\theta_2,\theta_3$. Equivalently, if $F$ is the affine transformation that brings $\triangle$ to the equilateral triangle (the unit circle in the triangle distance), then there is an edge $uv$ in the $\text{TD}_{\theta_1,\theta_2}$-Delaunay graph of a set $S\subseteq\mathbb{R}^2$ exactly when $F(u)F(v)$ is an edge of the TD-Delaunay graph of $F(S)$. This alternative definition immediately leads to a local routing strategy for the $\text{TD}_{\theta_1,\theta_2}$-Delaunay graph of a point set $S$: use the existing routing algorithm from \cite{DBLP:journals/siamcomp/BoseFRV15} on the TD-Delaunay graph of $F(S)$. In Section \ref{sec:compare}, we show that this approach is not optimal. 

Bonichon et al. \cite{DBLP:conf/wg/BonichonGHI10} showed that the TD-Delaunay graph corresponds to the half-$\theta_6$-graph. Analogous to the half-$\theta_6$-graph, Lubiw and Mondal \cite{DBLP:journals/jgaa/LubiwM19} define the $3$-sweep graph, which directly corresponds to the $\text{TD}_{\theta_1,\theta_2}$-Delaunay graph. The 3-sweep graph $G$ gets its name from an alternative, yet equivalent, construction. For each vertex $u$ and each positive cone $C_{u,i}$, include in $G$ the edge to the \textit{nearest} vertex $v\in C_{u,i}$. By \textit{nearest}, we mean that the triangle $T^{u,v}$ is minimal among $\{T^{u,v'}\mid v'\in C_{u,i}\}$. In this way, one can picture the leading edge $\tau^{u,v}_{i-1} \tau^{u,v}_{i+1}$ \textit{sweeping} through cone $C_{u,i}$. Throughout the paper, we assume that no two points lie on a line parallel to a cone boundary. This ensures that each vertex has at most one neighbour in each positive cone.

One desirable property of paths is angle monotonicity. A path is angle monotone with width $\alpha$ if the vector of each edge on the path lies in a cone with apex angle $\alpha$. Lubiw and Mondal show that the 3-sweep graph has certain angle-monotone properties which are used to upper bound the spanning ratio, see Observation \ref{obs:angleMono}.
\begin{obs}\label{obs:angleMono}
    An angle monotone path from $s$ to $t$ with width $\alpha$ has length at most $\frac{|st|}{\cos(\alpha/2)}$ \cite{DBLP:journals/jgaa/LubiwM19}. 
\end{obs}
In \cite{DBLP:journals/jgaa/LubiwM19}, Lubiw and Mondal also define a $k$-layered 3-sweep graph by combining $k$ copies of rotated 3-sweep graphs, and provide a local routing algorithm that finds angle monotone paths in $k$-layered 3-sweep graphs. Note that since $k$ is at least $4$, their routing algorithm does not apply to $\text{TD}_{\theta_1,\theta_2}$-Delaunay graphs.

\subsection{Our Contributions}

In Section \ref{sec:span}, we prove that $\frac{1}{\sin(\theta_1/2)}$ is a lower bound on the spanning ratio of $\text{TD}_{\theta_1,\theta_2}$-Delaunay graphs which matches the best known upper bound \cite{DBLP:journals/jgaa/LubiwM19}. Then in Section \ref{sec:route}, we provide a lower bound on the routing ratio by showing that there exist $\text{TD}_{\theta_1,\theta_2}$-Delaunay graphs for which the routing ratio of any $k$-local routing algorithms is at least as large as $C(\theta_1,\theta_2)$. Then, we show that our lower bound is tight by providing an online local routing algorithm for $\text{TD}_{\theta_1,\theta_2}$-Delaunay graphs with a routing ratio of $C(\theta_1,\theta_2)$. Finally, in Section \ref{sec:compare}, we compare our optimal routing algorithm to the previously best-known approach to routing in $\text{TD}_{\theta_1,\theta_2}$-Delaunay graphs.  

\section{Spanning Ratio}\label{sec:span}
We present a lower bound in the following proposition.
\begin{prop}\label{prop:TLBSR}
    There exists a set of points $S\subseteq\mathbb{R}^2$ such that the $\text{TD}_{\theta_1,\theta_2}$-Delaunay graph of $S$ has a spanning ratio of exactly $\frac{1}{\sin(\theta_1/2)}-\epsilon$ for any $\epsilon>0$.
\end{prop}
\begin{proof}
    We will construct a point set $S=\{ a,b,\tau_1,\tau_2, \tau_3 \}$ such that $d_G(a,b)$ approaches $\frac{|ab|}{\sin(\theta_1/2)}$, where $G$ is the $\text{TD}_{\theta_1,\theta_2}$-Delaunay graph of $S$. See Figure \ref{fig:SRLB}. Place two points $a,b$ outside $\triangle$ each at a distance $\frac{\min(|\tau_1\tau_2|,|\tau_1\tau_3|)}{2}$ from $\tau_1$, with $a$ arbitrarily close to $\tau_1\tau_2$ and $b$ arbitrarily close to $\tau_1\tau_3$. By construction of $S$, $G$ has edges  $\tau_1\tau_2$, $\tau_2\tau_3$, $\tau_1\tau_3$,    $\tau_1a$, $\tau_2a$, $\tau_1b$ and $\tau_3b$.

\begin{figure}[ht!]
    \centering
    \includegraphics[page=1, scale=0.7]{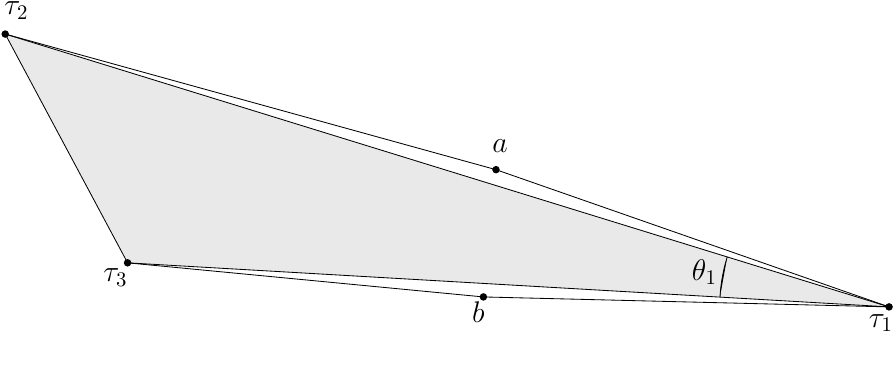}
    \caption{The shortest path from $a$ to $b$ passes through $\tau_1$ in the $\text{TD}_{\theta_1,\theta_2}$-Delaunay graph $G$ of the point set $\{ a,b,\tau_1,\tau_2, \tau_3 \}$. As $a,b$ get closer to $\triangle$, then $d_G(a,b)$ approaches $\frac{|ab|}{\sin(\theta_1/2)}$.\label{fig:SRLB}}
\end{figure}

The shortest path in $G$ from $a$ to $b$ passes through $\tau_1$, meaning the spanning ratio is at least $\frac{|a\tau_1| +|\tau_1b|}{|ab|}$. This value can be made arbitrarily close to $\frac{1}{\sin(\theta_1/2)}$ as $a$ and $b$ move closer to the boundary of $\triangle$. While this point set may not be in general position, the vertices can be perturbed to satisfy the general position constraint.
\end{proof}

The upper bound of $\frac{1}{\sin(\theta_1/2)}$ follows from Lemma 6 of \cite{DBLP:journals/jgaa/LubiwM19} by Lubiw and Mondal. 

\section{Local Routing}\label{sec:route}

Local routing has been studied in many contexts, and in Section \ref{sec:compare}, we will show that the known routing algorithms do not give optimal results in $\text{TD}_{\theta_1,\theta_2}$-Delaunay graphs. In this section, we provide an optimal local routing algorithm. Our approach is to generalize the algorithm from \cite{DBLP:journals/siamcomp/BoseFRV15}, leading to our algorithm (refer to Algorithm \ref{alg:LocalTriRoute}). The key algorithmic insight lies in the threshold for making decisions in routing. Each decision is carefully made to reduce the total path length. The goal of this section is to prove the following theorem.
\begin{theorem}\label{thm:TriRR}
    The routing ratio of the $\text{TD}_{\theta_1,\theta_2}$-Delaunay graph is at most $C(\theta_1,\theta_2)$. Furthermore, this bound is tight in the worst case.
\end{theorem}

We will start with the following proposition:
\begin{prop}\label{prop:LBRR}
    Let $k$ be a positive integer. Every $k$-local routing algorithm for $\text{TD}_{\theta_1,\theta_2}$-Delaunay graphs must have a routing ratio at least $C(\theta_1,\theta_2)-\epsilon$ for any $\epsilon>0$.
\end{prop}

\begin{proof}
We will construct two vertex sets $S_1$ and $S_2$ and refer to their corresponding $\text{TD}_{\theta_1,\theta_2}$-Delaunay graphs as $G_1$ and $G_2$. Importantly, the $k$-neighbourhoods of $G_1$ and $G_2$ around the start vertices $s$ are identical, however the rest of the graphs will be vastly different. In this way any algorithm that performs well for one graph will not for the other. These are analogous to the constructions of Figure 12 in \cite{DBLP:journals/siamcomp/BoseFRV15}.
\begin{figure}[ht!]
    \centering
    \includegraphics[page=20, scale=0.7]{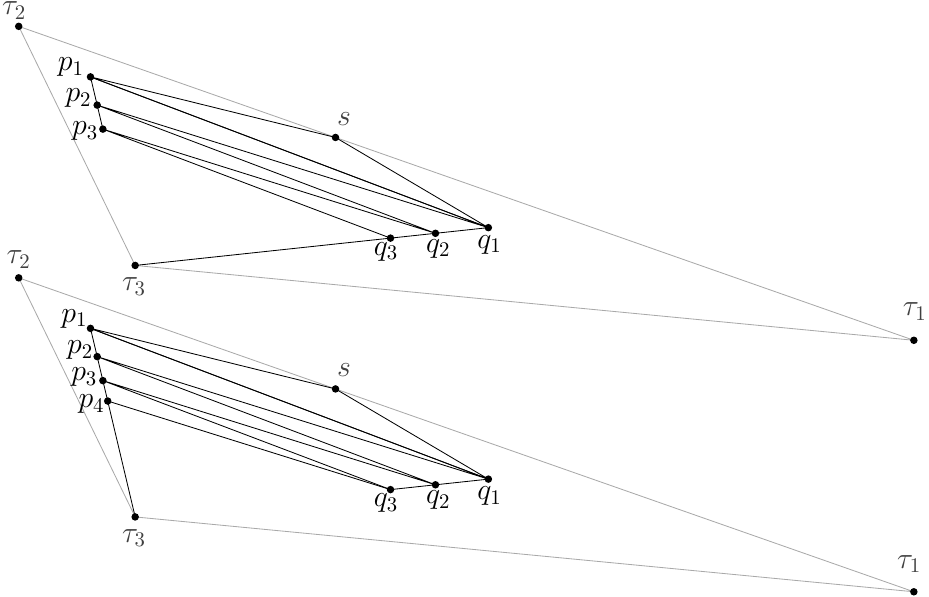}
    \caption{The $\text{TD}_{\theta_1,\theta_2}$-Delaunay graphs $G_1$ and $G_2$ constructed for the lower bound of $k$-local routing from $s$ to $\tau_3$. In this example, $k=3$.\label{fig:kLocalInstances}}
\end{figure}
Assume $j=3$ maximizes the expression of $C(\theta_1,\theta_2)$. Let $s$ be on $\tau_1\tau_2$. Place $p_1$ inside $C_{s,1}\cap C_{\tau_2,2}$ arbitrarily close to $\tau_2$, then place $q_1$ in $C_{s,2}\cap C_{p_1,2}\cap C_{\tau_1,1}$ arbitrarily close to $\tau_1$. Next, place $p_2$ on segment $\tau_3 p_1$ in cone $C_{q_1,1}$, arbitrarily close to $p_1$. Next, for $i=2,...,k$, place $q_i$ such that triangle $\tau_3,p_i,q_i$ is similar to $\tau_3,p_1,q_1$, and place $p_{i+1}$ such that triangle $\tau_3,p_{i+1},q_i$ is similar to $\tau_3,p_2,q_1$. Finally, let $S_1=\{s,p_1,...,p_k, q_1,...,q_k, \tau_3\}$ and $S_2=\{s,p_1,...,p_k,p_{k+1}, q_1,...,q_k, \tau_3\}$. This construction ensures that $G_1$ contains the edges $sp_1,sq_1,p_1q_1,p_{i-1}p_i,q_{i-1}q_i,q_{i-1} p_{i},p_i q_{i},q_k\tau_3$ where $i\in\{2,...,k\}$. On the other hand, $G_2$ contains the edges $sp_1, sq_1, p_1q_1, p_{i-1}p_i, q_{i-1}q_i$, $q_{i-1} p_{i}, p_i q_{i},p_{k}p_{k+1}, q_{k}p_{k+1},p_{k+1}\tau_3$ where $i\in\{2,...,k\}$. Importantly, $G_1$ does not contain the edge $\tau_3p_{k}$ because $q_k$ is the closest neighbour to $\tau_3$ in the cone $C_{\tau_3,3}$. Similarly, $G_2$ does not contain the edge $\tau_3q_{k}$ because $p_{k+1}$ is the closest neighbour to $\tau_3$ in the cone $C_{\tau_3,3}$. Similarly, the edges $p_k\tau_3$ and $q_k\tau_3$ do not exist in $G_1$ and $G_2$, respectively, since $\tau_3$ is in a negative cone of $p_k$ and $q_k$.

Since the $k$-neighbourhood of $s$ in $G_1$ and $G_2$ is $\{s,p_1,...,p_k,q_1,...,q_k\}$, then any algorithm routing from $s$ to $\tau_3$ will choose the same first vertex ($p_1$ or $q_1$) in $G_1$ and $G_2$. Moreover, $\tau_3$ only has one neighbour in each graph, so any path from $s$ to $\tau_3$ must pass through $q_k$ in $G_1$ and through $p_{k+1}$ in $G_2$. Then any algorithm that visits $p_1$ first will output a path from $s$ to $\tau_3$ in $G_1$ of length at least $|sp_1|+|p_1q_k|+|q_k\tau_3|$. On the other hand, any algorithm that chooses to visit $q_1$ first will output a path from $s$ to $\tau_3$ in $G_2$ having length at least $|sq_1|+|q_1p_{k+1}|+|p_{k+1}\tau_3|$. Since $p_1$ is arbitrarily close to $\tau_2$, $p_2$ is arbitrarily close to $p_1$, and $q_1$ is arbitrarily close to $\tau_1$, then each $p_i$ is arbitrarily close to $\tau_2$ and each $q_i$ is arbitrarily close to $\tau_1$. Then, for any $\epsilon>0$, the routing ratio of any algorithm is at least
\begin{align*}
&\frac{\min(|s\tau_2|+|\tau_2\tau_1|+|\tau_1\tau_3|,|s\tau_1|+|\tau_1\tau_2|+|\tau_2\tau_3|)}{|s\tau_3|}-\epsilon\\
&=\frac{|s\tau_1|}{|s\tau_3|}+\frac{|s\tau_2|}{|s\tau_3|}+\min\Big(\frac{|s\tau_2|}{|s\tau_3|}+\frac{|\tau_1\tau_3|}{|s\tau_3|}, \frac{|s\tau_1|}{|s\tau_3|}+\frac{|\tau_2\tau_3|}{|s\tau_3|}\Big)-\epsilon.
\end{align*}
Finally, we obtain $C(\theta_1,\theta_2)-\epsilon$ by the law of sines in triangles $s\tau_2\tau_3$ and $\tau_1 s\tau_3$, where angle $\alpha:=\angle \tau_2\tau_3 s$, since
\begin{align*}
    \frac{|s\tau_2|}{\sin(\alpha)}=\frac{|s\tau_3|}{\sin(\theta_2)}=\frac{|\tau_2\tau_3|}{\sin(\pi-\alpha-\theta_2)},\quad\text{and}\quad
    \frac{|s\tau_1|}{\sin(\theta_3-\alpha)}=\frac{|\tau_1\tau_3|}{\sin(\alpha+\theta_2)}=\frac{|s\tau_3|}{\sin(\theta_1)}.
\end{align*}
\end{proof}

\subsection{Local Routing Algorithm}

In this section, we present Algorithm \ref{alg:LocalTriRoute} which is a $1$-local, $0$-memory routing algorithm for $\text{TD}_{\theta_1,\theta_2}$-Delaunay graphs. It is generalized from the routing algorithm by Bose et al. \cite{DBLP:journals/siamcomp/BoseFRV15}.
Let $s$ be the start vertex, $t$ be the target vertex, and $p$ be the current vertex. At each step of Algorithm \ref{alg:LocalTriRoute}, the next vertex is chosen based on the four cases \ref{case:a}, \ref{case:b}, \ref{case:c}, or \ref{case:d}. To ease notation for cases \ref{case:b}, \ref{case:c}, and \ref{case:d}, we will define the left, middle, and right regions of $p$: $X_L$, $X_M$, and $X_R$ respectively, pictured in Figure \ref{fig:Xregions}. When $t$ lies in a negative cone $\overline{C_{p,i}}$, then let $X_L:= C_{p,i-1}\cap T^{p,t}$, $X_R:= C_{p,i+1}\cap T^{p,t}$, and $X_M:= \overline{C_{p,i}}\cap T^{p,t}$.

\begin{figure}
    \centering
    \includegraphics[page=6, scale=0.7]{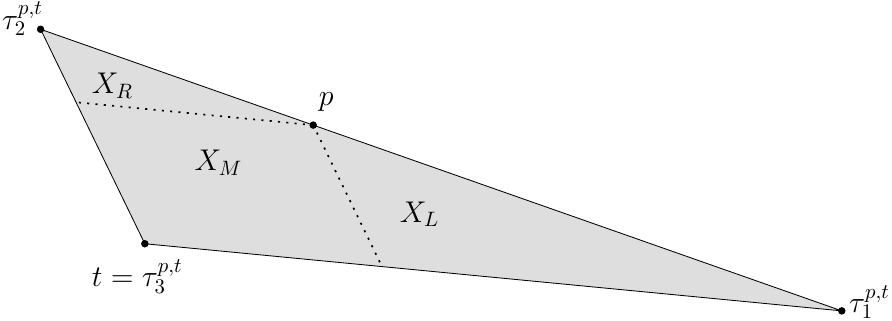}
    \caption{ $T^{p,t}$ is the smallest scaled translate of $\triangle$ with $p$ and $t$ on its boundary.}\label{fig:Xregions}
\end{figure}

\RestyleAlgo{ruled}

\begin{algorithm}
\caption{Local Routing algorithm in $\text{TD}_{\theta_1,\theta_2}$-Delaunay graph $G$}\label{alg:LocalTriRoute}
\KwData{Two points $s,t\in S$}
\KwResult{Path in $G$ from $s$ to $t$}
$p\gets s$\;
\While{$p \neq t$}{
Choose the next vertex $v$ based on the following cases, then set $p\gets v$
\begin{enumerate}[label=(\roman*)]
    \item\label{case:a} Case: \textit{$t$ lies in a positive cone $C_{p,i}$.}
    
    Follow the unique edge $pv$ in $C_{p,i}$.

    \item\label{case:b} Case: \textit{$t$ lies in a negative cone $\overline{C_{p,i}}$, and both regions $X_L$ and $X_R$ are empty } 

    Let $j\in\{1,-1\}$ minimize $|p\tau^{p,t}_{i+j}|+|\tau^{p,t}_{i+j}t|$. Choose the neighbour in $X_M$ closest to $C_{p,i+j}$ in cyclic order about $p$.
    
    \item\label{case:c} Case: \textit{$t$ lies in a negative cone $\overline{C_{p,i}}$, and only one region of $\{X_L,X_R\}$ is empty. } 

    If $p$ has neighbours in $X_M$, choose the neighbour $v$ closest to the empty region in cyclic order about $p$. Otherwise, choose the unique neighbour in the non-empty region.

    \item\label{case:d} Case: \textit{$t$ lies in a negative cone $\overline{C_{p,i}}$, and neither $X_L$ nor $X_R$ is empty.} 

    If $p$ has neighbours in $X_M$, choose an arbitrary one. Otherwise, let $j\in\{1,-1\}$ minimize $|p\tau^{p,t}_{i+j}|+|\tau^{p,t}_{i-j}t|$, and choose $v$ in $C_{p,i+j}$.
\end{enumerate}}
\end{algorithm}

In short, the algorithm prefers to route in the region towards $t$, however when this is not possible, it stays close to a neighbouring empty region or the side that minimizes a possible detour. Now we will prove the following upper bound:

\begin{prop}\label{prop:RRUB}
    Let $s,t$ be two vertices in a $\text{TD}_{\theta_1,\theta_2}$-Delaunay graph $G$. When $t$ is in a negative cone of $s$, then the path $P_{s,t}$ output by Algorithm \ref{alg:LocalTriRoute} from $s$ to $t$ in $G$ has ratio $\frac{|P_{s,t}|}{|st|}$ at most $C(\theta_1,\theta_2)$. When $t$ is in a positive cone of $s$, then $\frac{|P_{s,t}|}{|st|}$ is at most $\frac{1}{\sin(\theta_1/2)}$.
\end{prop}
Notice that when the angles $\theta_1,\theta_2,\theta_3$ are all equal to $\frac{\pi}{3}$, then $C(\theta_1,\theta_2)$ in Proposition \ref{prop:LBRR} for routing in a negative cone reaches a maximum of $5/\sqrt{3}$ when $\alpha=\frac{\pi}{6}$, matching the bound from \cite{DBLP:journals/siamcomp/BoseFRV15}. Furthermore, the expression for routing in a positive cone matches the spanning ratio.

\begin{proof}
We will bound the path chosen by Algorithm \ref{alg:LocalTriRoute} by defining a potential for each vertex along a path and showing that at each step, the potential drops by at least the length of the chosen edge. We define the potential as follows, depicted in Figure \ref{fig:potential}. 
\begin{itemize}
    \item Case \ref{case:a}: $\Phi(p,t):=\max\limits_{j=\pm1}(|p\tau^{p,t}_{i+j}|+|\tau^{p,t}_{i+j}t|)$.

    \item Case \ref{case:b}: $\Phi(p,t):=\min\limits_{j=\pm1}(|p\tau^{p,t}_{i+j}|+|\tau^{p,t}_{i+j}t|)$.

    \item Case \ref{case:c}:  $\Phi(p,t):=|p\tau^{p,t}_{i+j}|+|\tau^{p,t}_{i+j}t|$ where the empty region ($X_L$ or $X_R$) is $C_{p,i+j}\cap T^{p,t}$.

    \item Case \ref{case:d}:  $\Phi(p,t):=\min\limits_{j=\pm1}(|p\tau^{p,t}_{i+j}|+|\tau^{p,t}_{i+j}\tau^{p,t}_{i-j}|+|\tau^{p,t}_{i-j}t|)$.
\end{itemize}

\begin{figure}[hbt!]
    \centering
    \includegraphics[page=7, scale=0.7]{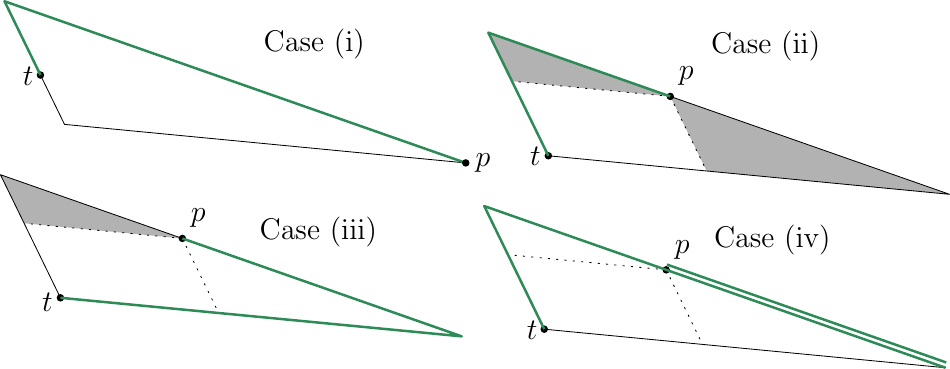}
    \caption{The potential is given by the green path. In this example, $i$ takes values $1,3,3,3$ for cases \ref{case:a}, \ref{case:b}, \ref{case:c}, \ref{case:d} respectively. The grey regions are empty.}\label{fig:potential}
\end{figure}

Now, we will show that in each case of Algorithm \ref{alg:LocalTriRoute}, the length of each chosen edge is less than the drop in potential. More precisely, we want to show $|pv| + \Phi(v,t)\leq \Phi(p,t)$ for cases \ref{case:a},\ref{case:b},\ref{case:c}, and \ref{case:d}.

Suppose the current vertex is $p$ and the case is \ref{case:a}, as can be seen in Figure \ref{fig:CaseA}. Then after an edge $pv$ is chosen, the current vertex will proceed to $v$ and the case will be either \ref{case:a}, \ref{case:b}, or \ref{case:c}. Case \ref{case:d} is not possible when $t$ lies in a negative cone of $v$ because at least one of the regions of $v$ is empty. Then the next potential, $\Phi(v, t)$, passes through some vertex $\tau^{v,t}_{i+k}$ for $k=\pm1$. We have
\begin{align*}
    |pv| + \Phi(v,t) &\leq (|p\tau^{p,v}_{i+k}|+ |\tau^{p,v}_{i+k}v|)+ (|v\tau^{v,t}_{i+k}|+|\tau^{v,t}_{i+k}t|) \\
    &= (|p\tau^{p,v}_{i+k}|+|v\tau^{v,t}_{i+k}|) + ( |\tau^{p,v}_{i+k}v|+|\tau^{v,t}_{i+k}t|) \\
    &= |p\tau^{p,t}_{i+k}|+|\tau^{p,t}_{i+k}t|  \\
    &\leq \max\limits_{j=\pm1}(|p\tau^{p,t}_{i+j}|+|\tau^{p,t}_{i+j}t|) =\Phi(p,t) 
\end{align*}
\begin{figure}[hbt!]
    \centering
    \includegraphics[page=15, scale=0.7]{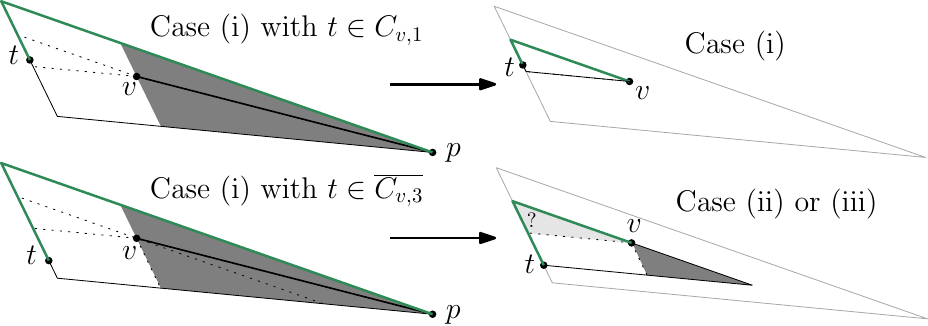}
    \caption{Bounding the potential in case \ref{case:a} since $t$ lies in $C_{p,1}$. The dark grey regions are empty. }\label{fig:CaseA}
\end{figure}

Next, suppose the current vertex is $p$ and the case is \ref{case:b}, depicted in Figure \ref{fig:CaseB}. Let $j$ minimize the expression from $\Phi(p,t)$. Notice that since we choose the edge closest in cyclic order about $p$ to the region $C_{p,i+j}\cap T^{p,t}$, then we can deduce that $v$ has no neighbours in its region $C_{v,i+j}\cap T^{v,t}$. Therefore once the current vertex proceeds to $v$, then the possible cases are only \ref{case:b} or \ref{case:c}. Then we have
\begin{align*}
    |pv| + \Phi(v, t) &\leq (|p\tau^{p,v}_{i+j}|+ |\tau^{p,v}_{i+j}v|)+ (|v\tau^{v,t}_{i+j}|+|\tau^{v,t}_{i+j}t|) \\
    &=(|p\tau^{p,v}_{i+j}|+ |v\tau^{v,t}_{i+j}|)+(|\tau^{p,v}_{i+j}v|+|\tau^{v,t}_{i+j}t|) \\
    &=|p\tau^{p,t}_{i+j}|+|\tau^{p,t}_{i+j}t| =\Phi(p , t) 
\end{align*}
\begin{figure}[hbt!]
    \centering
    \includegraphics[page=14, scale=0.7]{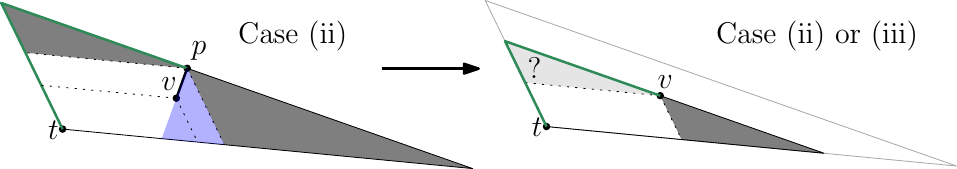}
    \caption{Case \ref{case:b} when $t$ lies in $\overline{C_{p,3}}$ and $j=-1$. The dark grey and blue regions are empty.}\label{fig:CaseB}
\end{figure}

Now suppose $p$ is the current vertex, the case is \ref{case:c}, and $C_{p,i+j}\cap T^{p,t}$ is the empty region, shown in Figure \ref{fig:CaseC}. If $X_M$ is not empty, then the choice of closest neighbour $v$ to $C_{p,i+j}\cap T^{p,t}$ guarantees that the corresponding region $C_{v,i+j}\cap T^{v,t}$ of $v$ is also empty. Likewise, if $X_M$ is empty, then choosing the unique neighbour in $C_{p,i-j}\cap T^{p,t}$ again guarantees that $v$ has an empty region $C_{v,i+j}\cap T^{v,t}$. Either way, once the current vertex continues to $v$, then the case must be either \ref{case:b} or \ref{case:c}. Then the exact same sequence of inequalities as from case \ref{case:b} completes the argument. 

\begin{figure}[hbt!]
    \centering
    \includegraphics[page=13, scale=0.7]{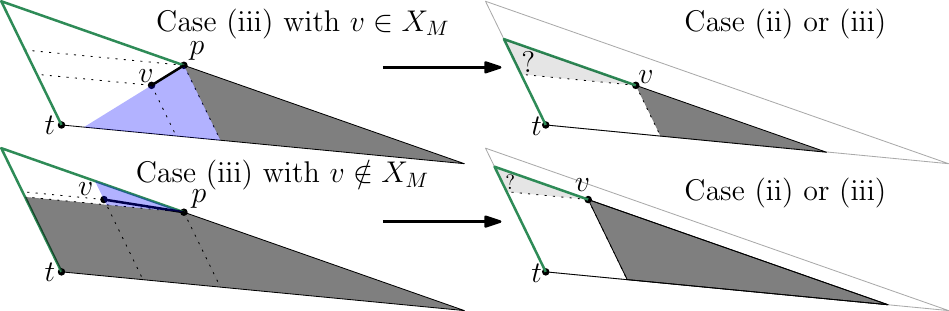}
    \caption{Case \ref{case:c} when $t$ lies in $\overline{C_{p,3}}$ and $j=-1$. The dark grey and blue regions are empty.}\label{fig:CaseC}
\end{figure}

Finally, we move on to case \ref{case:d}, where $p$ is the current vertex. After choosing the next vertex $v$, the next possible cases are \ref{case:b}, \ref{case:c}, or \ref{case:d}. Let $j=\pm1$ minimize the expression of $\Phi(p, t)$. In any case, we have $\Phi(v,t)\leq (|v\tau^{v,t}_{i+j}|+|\tau^{v,t}_{i+j}\tau^{v,t}_{i-j}|+|\tau^{v,t}_{i-j}t|)$ by the triangle inequality. When $v\in X_M$, then we use the following inequalities to prove the claim, also shown in Figure \ref{fig:CaseDIneq}.

\begin{enumerate}
    \item $|pv|\leq |p\tau^{v,p}_{i-j}|+|\tau^{v,p}_{i-j}v|$ by triangle inequality
    \item $|\tau^{v,t}_{i-j}t|+|\tau^{v,p}_{i-j}v| = |\tau^{p,t}_{i-j}t|$ by projection
    \item $|p\tau^{v,p}_{i-j}|\leq |\tau^{v,p}_{i+j}\tau^{v,p}_{i-j}|$ since $p$ lies on $\tau^{v,p}_{i+j}\tau^{v,p}_{i-j}$
    \item $|\tau^{v,t}_{i+j} \tau^{v,t}_{i-j}|+|\tau^{v,p}_{i+j}\tau^{v,p}_{i-j}| = |\tau^{p,t}_{i+j} \tau^{p,t}_{i-j}|$ by translation and projection
    \item $|v\tau^{v,t}_{i+j}| \leq |p\tau^{t,p}_{i+j}|$ by projection
\end{enumerate}
    
When $v$ is not in $X_M$, then let $u$ be the intersection of $p\tau^{p,v}_i$ and $\tau^{t,v}_{i-j}\tau^{t,v}_{i+j}$. The following inequalities suffice to prove the claim.
\begin{enumerate}
    \item $|pv|\leq |pu|+|uv|$ by triangle inequality
    \item $|uv|+|v\tau^{v,t}_{i+j} | \leq |p\tau^{p,t}_{i+j}|$ by projection
    \item $|\tau^{v,t}_{i-j}t|+|pu| = |\tau^{p,t}_{i-j}t|$ by projection
    \item $|\tau^{v,t}_{i+j} \tau^{v,t}_{i-j}| \leq |\tau^{p,t}_{i+j} \tau^{p,t}_{i-j}|$ by projection
\end{enumerate}

\begin{figure}[hbt!]
    \centering
    \includegraphics[page=12, scale=0.7]{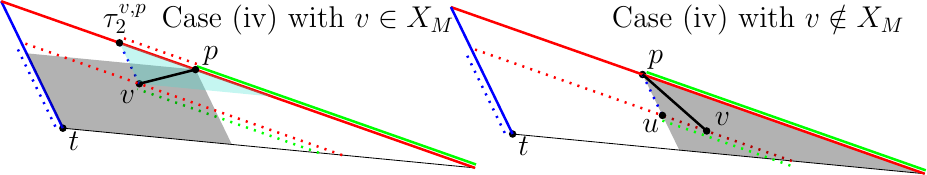}
    \caption{Bounding the potential in case \ref{case:d} when $t$ lies in $\overline{C_{p,3}}$, with $i=3, j=1$. The dotted paths representing $|pv|+\Phi(v,t)$ are shorter than the corresponding solid paths of $\Phi(p,t)$. The grey region contains $v$, and the blue triangle is $T^{v,p}$.}\label{fig:CaseDIneq}
\end{figure}

Since $\Phi(t,t)=0$, then the path from $s$ to $t$ output by Algorithm \ref{alg:LocalTriRoute} can have length at most $\Phi(s,t)$. When $t$ is in a positive cone of $s$, then the potential is defined using case \ref{case:a}. The corresponding path $p\tau^{p,t}_{i+j}+\tau^{p,t}_{i+j}t$ is $\pi-\theta_{i+j}$ monotone, then the routing ratio in such a case can be at most $\frac{1}{\sin(\theta_1/2)}$ by Observation \ref{obs:angleMono}. 

On the other hand, when $t$ is in a negative cone of $s$, there are three possible cases: \ref{case:b}, \ref{case:c} or \ref{case:d}. The triangle inequality tells us that $\Phi(s,t)$ is largest in case \ref{case:d}. Then, similar to the proof of Proposition \ref{prop:LBRR}, the routing ratio is bounded by $C(\theta_1,\theta_2)$ using the law of sines.
\end{proof}
Finally, Theorem \ref{thm:TriRR} is a consequence of Propositions \ref{prop:LBRR} and \ref{prop:RRUB} since Algorithm \ref{alg:LocalTriRoute} is $1$-local.

\subsection{Comparison to known routing algorithms}\label{sec:compare}

In this subsection, we show that currently known local routing algorithms when applied on the $\text{TD}_{\theta_1,\theta_2}$-Delaunay graph are suboptimal. Firstly, note that by using a stretch factor upper bound from Section \ref{sec:span}, we can apply the technique of Bose and Morin \cite{DBLP:journals/tcs/BoseM04} to obtain a local routing algorithm that finds a path between any two vertices with length at most $9$ times the stretch factor, which is not optimal. Another approach is to route in $\text{TD}_{\theta_1,\theta_2}$-Delaunay graphs by combining the algorithm of Bose et al. \cite{DBLP:journals/siamcomp/BoseFRV15} with an affine transformation. 
When $\triangle$ is the equilateral triangle, then Algorithm \ref{alg:LocalTriRoute} simplifies to the standard TD-Delaunay routing algorithm from \cite{DBLP:journals/siamcomp/BoseFRV15}. In this case, notice that the thresholds in cases \ref{case:b} and \ref{case:d} simplify so that $j$ refers to the corner $\tau^{p,t}_{i+j}$ nearest $p$. In other words, the decision threshold is the midpoint of the segment $\tau^{p,t}_{i+1}\tau^{p,t}_{i-1}$. We will analyze this standard TD-Delaunay routing algorithm when it is used on the affine transformation of a general $\text{TD}_{\theta_1,\theta_2}$-Delaunay graph. Since affine transformations preserve midpoints, then the decision threshold in case \ref{case:d} is also the midpoint of the segment $\tau^{p,t}_{i+1}\tau^{p,t}_{i-1}$. It is in this way that applying an affine transformation to the existing algorithm differs from our Algorithm \ref{alg:LocalTriRoute}. To see the difference in routing ratio of these two approaches, consider the construction of $G_1$ from Proposition \ref{prop:LBRR}. If we enforce $|s\tau_2|<|s\tau_1|$, then the path output by the affine transformation of the standard TD-Delaunay routing algorithm would choose to visit $p_1$ first. The routing ratio of this algorithm would therefore be at least 
\begin{align*}
    \frac{\sin(\theta_3-\alpha)}{\sin(\theta_1)}+\frac{\sin(\alpha)}{\sin(\theta_2)}+\frac{\sin(\alpha)}{\sin(\theta_2)}+\frac{\sin(\alpha+\theta_2)}{\sin(\theta_1)}-\epsilon
\end{align*}
where $\alpha:=\angle\tau_{2} \tau_{3} s$. For example, when $\theta_1=\frac{\pi}{6}$, $\theta_2=\frac{\pi}{5}$, and $\alpha=\frac{\pi}{3}$ then the routing ratio of the standard TD-Delaunay algorithm under an affine transformation is strictly more than $6.55$, whereas the optimal routing ratio is less than $6.52$ by Proposition \ref{prop:RRUB}.

\small
\bibliographystyle{abbrv}

\newpage
\end{document}